\documentclass[letterpaper, 10 pt, conference]{ieeeconf}  

\IEEEoverridecommandlockouts                              
\usepackage{graphicx}      
\makeatletter
\let\old@ssect\@ssect 
\makeatother

\usepackage{amssymb}
\usepackage{amsfonts}
\usepackage{mathrsfs}
\usepackage{amsmath}
\usepackage{textcomp}
\usepackage{float}
\usepackage[dvipsnames]{xcolor}
\usepackage{multirow}
\usepackage{csquotes}
\usepackage{epstopdf} 
\usepackage{gensymb}
\usepackage{breqn}
\usepackage{siunitx}
\usepackage{hyperref} 
\hypersetup{colorlinks=true,colorlinks,linkcolor={blue},citecolor={blue},urlcolor={red}} 
\makeatletter
\def\endfigure{\end@float}
\def\endtable{\end@float}
\makeatother

\usepackage{subcaption}

\usepackage[thicklines]{cancel}
\usepackage{amsthm}
\newtheoremstyle{boldtheorem}
  {\topsep}   
  {\topsep}   
  {\normalfont}
  {}          
  {\bfseries} 
  {.}         
  {.5em}      
  {}          

\theoremstyle{boldtheorem}

\newtheorem{thm}{Theorem}
\newtheorem{rem}{Remark}%
\newtheorem{cor}{Corollary}%
\newtheorem{assum}{Assumption}%
\usepackage{xcolor}
\usepackage[thinlines]{easytable}

\title{\LARGE \bf The Analysis and the Performance of the Parallel-Partial Reset Control System}

\author{Xinxin Zhang and S. Hassan HosseinNia*
\thanks{* Corresponding Author.}
\thanks{This work was not supported by any organization.}
\thanks{Xinxin Zhang and S. Hassan HosseinNia are with the Department of Precision and Microsystems Engineering (PME), Delft University of Technology, Mekelweg 2, 2828 CD, Delft, The Netherlands, with Emails:
        {\tt\small X.Zhang-15@tudelft.nl, S.H.HosseinNiaKani@tudelft.nl.}}%
}

\begin{document}

\maketitle
\thispagestyle{empty}
\pagestyle{empty}

\begin{abstract}
Reset controllers have demonstrated their effectiveness in enhancing performance in precision motion systems. To further exploiting the potential of reset controllers, this study introduces a parallel-partial reset control structure. Frequency response analysis is effective for the design and fine-tuning of controllers in industries. However, conducting frequency response analysis for reset control systems poses challenges due to their nonlinearities. We develop frequency response analysis methods for both the open-loop and closed-loop parallel-partial reset systems. Simulation results validate the accuracy of the analysis methods, showcasing precision enhancements exceeding 100\% compared to the traditional describing function method. Furthermore, we design a parallel-partial reset controller within the Proportional–Integral–Derivative (PID) control structure for a mass-spring-damper system. The frequency response analysis of the designed system indicates that, while maintaining the same bandwidth and phase margin of the first-order harmonics, the new system exhibits lower magnitudes of higher-order harmonics, compared to the traditional reset system. Moreover, simulation results demonstrate that the new system achieves lower overshoot and quicker settling time compared to both the traditional reset and linear systems.
\end{abstract}

\section{Introduction}
Reset controllers have demonstrated their ability to overcome the phase-gain limitation and waterbed effects inherent in linear control systems \cite{chen2019development}. The first reset element, known as the Clegg Integrator \cite{clegg1958nonlinear}, exhibits the same gain characteristic as the linear integrator. However, it offers a 52$\degree$ phase lead, as determined by describing function analysis. This property enables reset controllers to exhibit superior tracking and disturbance rejection benefits compared to their linear counterparts in various applications, including semiconductor manufacturing, chemical processes, and precision motion systems \cite{chen2019development, zheng2000experimental, hazeleger2016second, le2021passive, saikumar2019constant, palanikumar2018no}. 

Nevertheless, it is essential to acknowledge that reset control, while overcoming linear trade-offs, introduces higher-order harmonics due to its inherent nonlinearity, which can lead to undesirable performance, such as the limit cycle \cite{banos2012reset}. To mitigate the excessive influence of nonlinearity, various forms of reset controllers, such as partial reset controllers \cite{banos2012reset} and the ``Proportional-Integrator + Clegg Integrator (PI+CI)" configuration \cite{banos2007definition}, have been developed. To the best of our knowledge, the combination of linear and reset elements in parallel configuration is primarily explored within the context of the ``PI+CI" system. There remains significant potential to explore this structure's applicability with general reset elements.

Frequency response analysis is used to examine the phase and gain properties of a control system to sinusoidal inputs across various frequencies. It serves as an effective tool for engineers to design and fine-tune controllers to meet specific performance requirements in industrial applications. For designing and implementing the ``PI+CI" control system in the frequency domain, frequency response analysis is needed. Currently, classical Describing Function (DF) analysis is commonly utilized to analyze such the open-loop ``PI+CI" control system. However, DF analysis typically only considers the first-order harmonic, which will lead to analytical inaccuracies. Moreover, there is an absence of precise frequency response analysis methods tailored for closed-loop ``PI+CI" systems. 

Addressing these gaps in research motivation, this paper makes the following contributions:
\begin{itemize}
    \item We introduce a novel control system structure termed Parallel-Partial-Reset Control Systems (PP-RCSs), which expands upon the conventional parallel ``PI+CI" configuration to accommodate a wider variety of reset control systems.
    \item A Higher-Order Sinusoidal Input Describing Function (HOSIDF) to analyze the frequency response of open-loop and closed-loop PP-RCSs to sinusoidal reference inputs is developed, building upon our previous research for the traditional reset system.
    \item We design a PP-RCS for a mass-spring-damper system. The frequency response analysis shows that the PP-RCS significantly reduces the magnitudes of higher-order harmonics in the traditional reset system. Additionally, the PP-RCS exhibits superior transient response compared to traditional reset and linear systems.
\end{itemize}
This paper is structured into six sections. In Section \ref{sec:preliminaries}, fundamental background of reset systems are provided. Section \ref{sec: PF} outlines the research problems addressed. Section \ref{sec:methodology1} elaborates on the frequency-domain analysis of PP-RCS, covering scenarios from open-loop to closed-loop. Simulation results illustrate a substantial accuracy improvement of the new analysis method compared to the DF method. In Section \ref{sec: conclusion}, a PP-RCS is designed on a precision motion stage. Simulations highlight the better performance achieved by PP-RCSs compared to traditional reset control systems. Finally, the key findings of this paper are summarized, and future research prospects are outlined in Section \ref{sec: conclusion}.

\section{Background}
\label{sec:preliminaries}
\subsection{Definition of the Reset Controller}
Figure \ref{fig1:rcs} depicts the block diagram of a closed-loop reset system, comprising a reset controller denoted as $\mathcal{C}$, a linear controller $\mathcal{C}_2$, and a plant $\mathcal{P}$. The signals $r(t)$, $e(t)$, $u(t)$, and $y(t)$ are the reference input, error, control input, and system output signals respectively. 
\begin{figure}[htp]
	\centerline{\includegraphics[width=0.7\columnwidth]{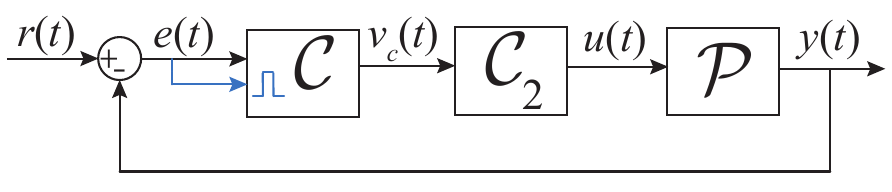}}
	\caption{The block diagram of the RCS, where the blue lines represent the reset action.}
	\label{fig1:rcs}
\end{figure}

The reset controller $\mathcal{C}$ is a linear time-invariant (LTI) system encompassing with a reset mechanism, which is triggered by the input signal $e(t)$. The state-space representative for the $\mathcal{C}$ is described as below:
\begin{equation} 
	\label{eq: State-Space eq of RC} 
	\mathcal{C} = \begin{cases}\dot{x}_r(t) = A_Rx_r(t) + B_Re(t), & t\notin J \\
		x_r(t^+) = A_\rho I x_r(t), & t\in J\\
		v_c(t) = C_Rx_r(t) + D_Re(t),\end{cases}
\end{equation} 
where $x_r(t) \in \mathbb{R}^{n_c}$ is the state of the RC and $n_c$ is the number of the states of the RC. $A_R$, $B_R$, $C_R$, $D_R$ are the state-space matrices of the base linear controller denoted as $\mathcal{C}_{bl}$, given by
\begin{equation}
\label{eq: RL}
C_{bl}(\omega) =  C_R(j\omega I-A_R)^{-1}B_R+D_R,\ (j = \sqrt{-1}).
\end{equation}
The reset controller $\mathcal{C}$ employs the ``zero-crossing law" as its reset mechanism, resetting specific states of $\mathcal{C}$ to zero when the reset triggered signal $e(t)$ crosses zero \cite{banos2012reset, li2010reset}. At the reset instant $t_i (i\in \mathbb{Z}^+)$, the reset action is activated when $e(t_i) = 0$. The set of reset instants $t_i$ is denoted as $J:=\{t_i|e(t_i)=0, i\in \mathbb{Z}^+\}$. The second equation describes the jump of the state $x_r$ from $t_i$ to $t_i^+$. $A_\rho$ represents the reset matrix, defined as:
\begin{equation}
\label{eq:gamma_def}
		A_\rho =
		\begin{bmatrix}
			A_{\rho\gamma} & \\
			& I_{n_l}
		\end{bmatrix}, A_{\rho\gamma} = \textrm{diag}(\gamma_1, \gamma_2,\cdots, \gamma_o, \cdots, \gamma_{n_r}).
\end{equation}
In \eqref{eq:gamma_def}, $\gamma_{o} \in (-1, 1)$, where $o\in\mathbb{Z}^+$, represents the reset value. The subscript $n_r$ denotes the number of reset states, while $n_l$ indicates the number of linear states, resulting in a total of $n_c = n_r + n_l$. For the scope of this study, our focus lies on reset controllers featuring a single reset state, specifically when $n_r=1$. Such reset controllers with single reset states are commonly encountered in literature, examples of which include the Clegg Integrator (CI), the First-order Reset Element (FORE) \cite{horowitz1975non}, the Second-order Reset Element (SORE) \cite{hazeleger2016second} with the reset applied to the first state, etc.

\subsection{Frequency Response Analysis for the Reset Controller}
A reset controller \eqref{eq: State-Space eq of RC} with an input signal of $e(t) = |E|\sin(\omega t + \angle E)$ is asymptotically stable with $2\pi/\omega$-periodic solution and convergent if and only if \cite{guo2009frequency}:
\begin{equation}
\label{eq:open-loop stability}
    |\lambda (D_Re^{A_R}\delta)|<1,\ \forall \delta \in \mathbb{R}^+.
\end{equation} 
Since the frequency response analysis necessitates system stability and convergence, we introduce the following assumption:
\begin{assum}
\label{open-loop stability}
The reset controller \eqref{eq: State-Space eq of RC} with input $e(t) = |E|\sin(\omega t + \angle E)$ is assumed to meet the condition in \eqref{eq:open-loop stability}.
\end{assum}
Consider a reset controller $\mathcal{C}$ \eqref{eq: State-Space eq of RC} with input signal $e(t) = |E|\sin(\omega t+\angle E_1)$, satisfying Assumption \ref{open-loop stability}. Utilizing the ``Virtual Harmonic Generator" \cite{heinen2018frequency}, the input signal $e(t)$ generates $n\in\mathbb{N}$ harmonics $e_{1n}(t) = |E|\sin(n\omega t+n\angle E_1)$. There are $n$ harmonics in the output signal $v_c(t)$, given by $v_c(t)=\sum_{n=1}^{\infty}v_c^n(t)$. The function $H_n(\omega)$ represents the transfer function from $e_{1n}(t)$ to $v_c^n(t)$, which is provided in \cite{heinen2018frequency, saikumar2021loop} and given by
\begin{equation}
	\label{eq: HOSIDF}
	\resizebox{1\hsize}{!}{$
\begin{aligned}
    H_n(\omega)=
    \begin{cases}
			C_R(j\omega I-A_R)^{-1}(I+j\Theta _D(\omega))B_R+D_R, & \text{for}\ n=1\\
			C_R(jn\omega I-A_R)^{-1}j\Theta _D(\omega)B_R, & \text{for odd} \ n > 1\\
			0, & \text{for even } n \geqslant 2
   \end{cases}
\end{aligned}
   		$}
\end{equation}
with
\begin{equation}
\label{Theta_D}
	\begin{aligned}
		\Lambda(\omega) &= \omega ^2I+{A_R}^2,\\
		\Delta(\omega) &= I+e^{(\frac{\pi}{\omega}A_R)},\\
		\Delta _r(\omega) &= I+A_{\rho}e^{(\frac{\pi}{\omega}A_R)},\\
		\Gamma _r(\omega) &= \Delta _r^{-1}(\omega)A_{\rho}\Delta(\omega)\Lambda ^ {-1}(\omega),\\
		\Theta _D(\omega) &= \frac{-2\omega ^2}{\pi}\Delta(\omega)[\Gamma _r(\omega)-\Lambda ^ {-1}(\omega)].
	\end{aligned} 	
\end{equation}

\section{Problem Statement}
\label{sec: PF}
The block diagram of a Parallel-Partial Reset Control System (PP-RCS) is depicted in Fig. \ref{fig1:PP-RC system}. It comprises a control structure $\mathcal{H}_{p}$, a linear controller $\mathcal{C}_2$, and a plant $\mathcal{P}$.

The control element $\mathcal{H}_{p}$ integrates a reset controller $\mathcal{C}$ \eqref{eq: State-Space eq of RC} and its corresponding base-linear controller $\mathcal{C}_{bl}$ \eqref{eq: RL} configured in parallel, with proportions of $k_{rc}$ and $1-k_{rc}$, respectively, given by
\begin{equation}
\label{eq: Hpp}
\begin{aligned}
    \mathcal{H}_{p} &= \mathcal{H}_{rc} + \mathcal{H}_{bl},\\
    \mathcal{H}_{rc} &= k_{rc}\mathcal{C},\ k_{rc} \in [0,1] \in \mathbb{R},\\
    \mathcal{H}_{bl} &= (1-k_{rc})\mathcal{C}_{bl}.   
\end{aligned}
\end{equation}
\begin{figure}[h]
	\centerline{\includegraphics[width=0.8\columnwidth]{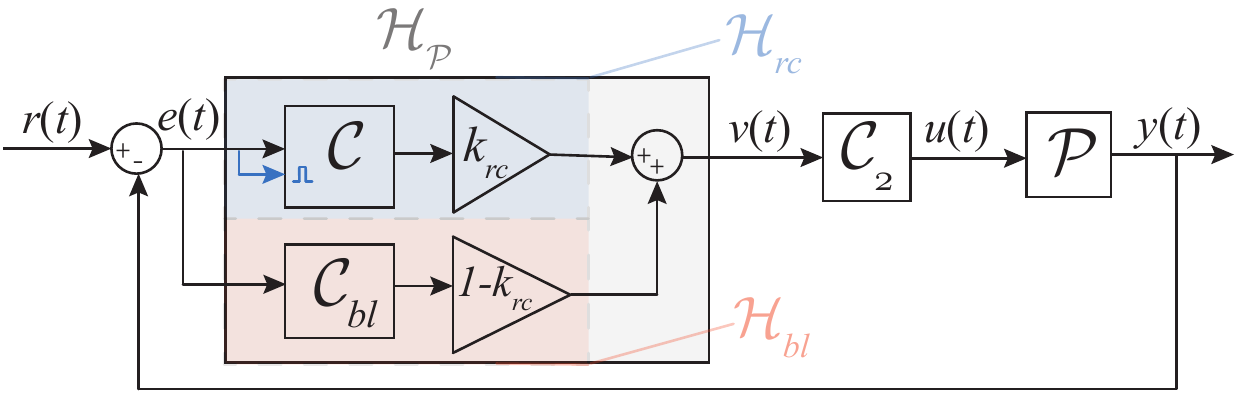}}
	\caption{The block diagram of the PP-RCS, where $r(t)$, $e(t)$, $v(t)$, $u(t)$, and $y(t)$ are the reference input, error, reset output, control input, and system output signals respectively. The blue lines represent the reset action.}
	\label{fig1:PP-RC system}
\end{figure}

The objective of this research is to develop a methodology for analyzing the frequency response of the PP-RCS to sinusoidal inputs, encompassing both its phase and gain characteristics across various frequencies. This analysis aims to provide insights into the steady-state behavior of the system and serve as a tool for designing the PP-RCS to meet specific specifications.

Let the set $\mathscr{L}_2$ consists of all measurable functions $f(\cdot):\mathbb{R}_+\to \mathbb{R}$ such that $\int_0^\infty|f(t)|^2 dt <\infty$, being the $\mathscr{L}_2$-norm $||\cdot||:\mathscr{L}_2\to\mathbb{R}_+$ defined by $||f|| = \sqrt{\int_0^\infty|f(t)|^2 dt}$. 

Note that the stability and convergence for reset systems can be ensured through thorough design practices, which are not the primary focus of this research. However, to ensure the existence of the frequency response of the closed-loop reset system, the following assumption is made to guarantee stability \cite{beker2004fundamental} and convergence \cite{dastjerdi2020frequency} of the closed-loop reset system:
\begin{assum}
\label{assum: stable}
The closed-loop PP-RCS is assumed to be $\mathscr{L}_2$-stable, the initial condition of the reset controller $\mathcal{C}$ is zero, there are infinitely many reset instants $t_i$ with $\lim_{t_i \to \infty} = \infty$, the input signal $r(t) = |R|\sin(\omega t)$ is a Bohl function \cite{barabanov2001bohl}, and there is no Zenoness behaviour.
\end{assum}

Under Assumption \ref{assum: stable}, as per \cite{pavlov2006uniform}, in Fig. \ref{fig1:PP-RC system}, $e(t)$, $v(t)$, $u(t)$, and $y(t)$ are periodic signals with the same fundamental frequency as $r(t)$ and can be expressed as follows:
\begin{equation}
\label{eq: e,y,u}
    \begin{aligned}
        e(t) &= \sum\nolimits_{n=1}^{\infty}e_n(t) = \sum\nolimits_{n=1}^{\infty} |E_n|\sin(n\omega t+\angle E_n),
    \end{aligned}
\end{equation}
where $\angle E_n\in(-\pi,\pi]$. The Fourier transform of $e_n(t)$ is $E_{n}(\omega)=\mathscr{F}[e_n(t)]$.

The closed-loop reset systems under a sinusoidal input with more than two reset instants possess excessive higher-order harmonics compared to systems with only two reset instants \cite{saikumar2021loop}. However, such effects can be mitigated through careful design considerations \cite{karbasizadeh2022band}. Moreover, practical reset control systems commonly employ systems with two reset instants \cite{banos2012reset}. Building on these observations, we introduce the following assumption:
\begin{assum}
\label{assum:2reset}
There are two reset instants in a SISO closed-loop PP-RCS with a sinusoidal reference input signal $r(t) = |R|\sin(\omega t)$ at steady-state, where the reset-triggered signal is $e_1(t)$. 
\end{assum}

Under Assumptions \ref{assum: stable} and \ref{assum:2reset}, the reset actions in the closed-loop SISO PP-RCS to the sinusoidal input $r(t) = |R|\sin(\omega t)$ occur when $e_{1}(t) = |E_1|\sin(\omega t + \angle E_1)= 0$. The set of reset instants for this closed-loop reset system is denoted as $J_m := \{t_m = (m\pi - \angle E_1)/\omega | m \in \mathbb{Z}^+\}$. 

Consider a SISO PP-RCS with $k_{rc}=1$, subjected to a reference input signal $r(t) = |R|\sin(\omega t)$, and under Assumptions \ref{assum: stable} and \ref{assum:2reset}. 

From \eqref{eq: HOSIDF}, The $n$-th open-loop transfer function of the PP-RCS with $k_{rc}=1$ denoted as $\mathcal{L}_n$ is given by
\begin{equation}
\label{Lo}
   \mathcal{L}_n(\omega)= H_n(\omega)\mathcal{C}_{2}(n\omega)\mathcal{P}(n\omega).
\end{equation}
Define 
\begin{equation}
\begin{aligned}
\label{eq: Cro defn}
\mathcal{C}_{nl}(\omega) &= H_1(\omega) -\mathcal{C}_{bl}(\omega), \\
\mathcal{C}_{nl}(n\omega) &= H_n(\omega), \text{ for }n>1.
\end{aligned}
\end{equation}
Then, the open-loop transfer function of the PP-RCS with $k_{rc}=1$ in \eqref{Lo} is separated into linear and nonlinear elements, given by:
\begin{equation}
\label{rcs:Lbl,Lnl}
    \begin{aligned}
\mathcal{L}_{bl}(n\omega) &= \mathcal{C}_{bl}(n\omega)\mathcal{C}_{2}(n\omega)\mathcal{P}(n\omega),\\
\mathcal{L}_{nl}(n\omega) &= \mathcal{C}_{nl}(n\omega)\mathcal{C}_{2}(n\omega)\mathcal{P}(n\omega),\\  
\mathcal{C}_{nl}(n\omega) &= C_R(jn\omega I-A_R)^{-1}j\Theta _D(\omega)B_R.
    \end{aligned}
\end{equation}
Utilizing the ``Virtual Harmonic Generator," the input signal $r(t)$ generates $n$ harmonics $r_n(t) = |R|\sin (n\omega t)$, with $R_n(\omega) = \mathscr{F}[r_n(t)]$. In our previous research \cite{zhang2022frequency}, we derived the $n^{\text{th}}$ sensitivity function $\mathcal{S}_n(\omega)$ for the PP-RCS with $k_{rc}=1$. It is defined as $\mathcal{S}_n(\omega)={E_n(\omega)}/{R_n(\omega)}$ and can be obtained from the open-loop transfer functions in \eqref{rcs:Lbl,Lnl}, given by:
 \begin{equation}
	 	\label{eq: sensitivity functions in CL1}
   \begin{aligned}
	\mathcal{S}_n(\omega)
	   &=\begin{cases}
	 	\frac{1}{1+\mathcal{L}_{o}(\omega)} , & \text{for } n=1\\
	 	-\frac{\Gamma(\omega)\mathcal{L}_{nl}(n\omega)|\mathcal{S}_{1}(\omega)|e^{jn\angle \mathcal{S}_{1}(\omega)}}{1+\mathcal{L}_{bl}(n\omega)} , & \text{for odd} \ n > 1\\
	0, & \text{for even} \ n \geqslant 2 \end{cases}       
   \end{aligned}
	 \end{equation}
  where
 \begin{equation}
\label{eq: Gamma for PP-RCS}
\begin{aligned}
\Delta_l(n\omega) &= (jn\omega I-A_R)^{-1}B_R,\\
\Delta_c(\omega) &= \lvert\Delta_l(\omega)\rvert \sin(\angle \Delta_l(\omega)),\\
\mathcal{L}_{o}(n\omega) &= \mathcal{L}_{bl}(n\omega) + \Gamma(\omega)\mathcal{L}_{nl}(n\omega),\\
\Psi_n(\omega) &= {|\mathcal{L}_{nl}(n\omega)|}/{|1+\mathcal{L}_{bl}(n\omega)|},\\
\Gamma(\omega) &= 1/[1-{\sum\nolimits_{n=3}^{\infty}\Psi_n(\omega)\Delta_c^n(\omega)}/{\Delta_c(\omega)}],\\
\Delta_c^n(\omega) &= -|\Delta_l(n\omega)| \sin(\angle\mathcal{L}_{nl}(n\omega)-\angle (1+\\
&\indent\indent\indent\indent\indent \indent \indent\indent \mathcal{L}_{bl}(n\omega)) + \angle \Delta_l(n\omega)).
\end{aligned}
\end{equation}   
The frequency response analysis for the PP-RCS with $k_{rc}\in [0, 1]$ is presented in Section \ref{sec:methodology1}.
\section{Frequency Response Analysis for the Parallel-Partial Reset Control System}
\label{sec:methodology1}
In this section, Theorem \ref{thm: Open-loop HOSIDF for PP-RCS} and Corollary \ref{thm: Closed-loop HOSIDF for PP-RCS} propose the Higher-order sinusoidal input describing function (HOSIDF) for open-loop and closed-loop PP-RCSs, respectively. 
\subsection{Frequency Response Analysis for the Open-loop and Closed-loop PP-RCS}
Consider a parallel-partial reset controller (PP-RC) $\mathcal{H}_p$ as depicted in Fig. \ref{fig1:PP-RC system}, operating under the sinusoidal input signal $e_1(t) = |E_1|\sin (\omega t + \angle E_1)$, where $\angle E_1\in[-\pi,\pi)$, and under the Assumption \ref{open-loop stability}. The steady-state output signal of $\mathcal{H}_p$ is denoted as $v_1(t)$, which comprises $n\in\mathbb{N}$ harmonics, represented as $v_1(t) = \sum\nolimits_{n=1}^{\infty}v_{1n}(t)$. The block diagram of $\mathcal{H}_p$ is illustrated in Fig. \ref{fig1:PP-RC system2}. By utilizing the ``Virtual Harmonic Generator" \cite{saikumar2021loop}, the input signal $e_1(t)$ generates $n$ harmonics, denoted as $e_{1n}(t) = |E_1|\sin (n\omega t + n\angle E_1)$.
   \begin{figure}[h]
	\centerline{\includegraphics[width=0.65\columnwidth]{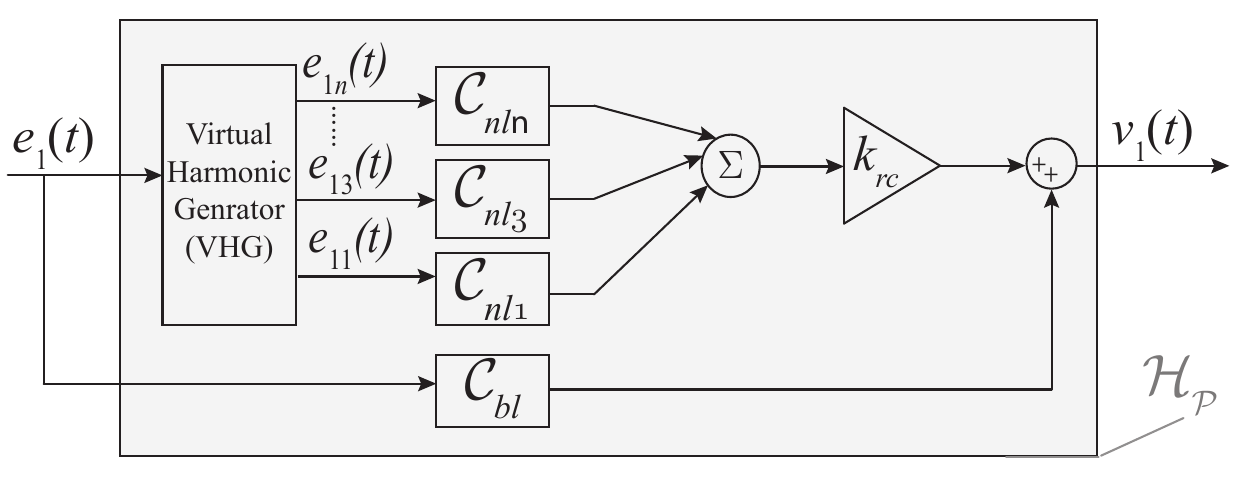}}
	\caption{The block diagram of the parallel-partial reset controller (PP-RC) $\mathcal{H}_p$.}
	\label{fig1:PP-RC system2}
\end{figure}
\begin{thm} (Frequency Response Analysis for the Open-loop PP-RCS)
\label{thm: Open-loop HOSIDF for PP-RCS}
The $n$-th Higher-order sinusoidal input describing function (HOSIDF) for $\mathcal{H}_p$, defined as $H_{p}^n(\omega)$, represents the transfer function from $E_{1n}(\omega) = \mathscr{F}[e_{1n}(t)]$ to $V_{1n}(\omega) = \mathscr{F}[v_{1n}(t)]$, and is given by:
    \begin{equation}
    \label{eq: H_{p}^n}
    H_{p}^n(\omega) = \frac{V_{1n}(\omega)}{E_{1n}(\omega)}=
        \begin{cases}
	 	\mathcal{C}_{bl}(\omega) + k_{rc} \mathcal{C}_{nl}(\omega), & \text{for } n=1\\
	 	k_{rc} \mathcal{C}_{nl}(n\omega), & \text{for odd} \ n > 1\\
	 	0, & \text{for even} \ n \geqslant 2 
        \end{cases}
    \end{equation}
    where $\mathcal{C}_{bl}(\omega)$ is given in \eqref{eq: RL} and $\mathcal{C}_{nl}(n\omega)$ is given in \eqref{rcs:Lbl,Lnl}.
\end{thm}
\begin{proof}
From \eqref{eq: Hpp} and \eqref{eq: Cro defn}, the $n$-th HOSIDF for $\mathcal{H}_{rc}$ is defined as
\begin{equation}
\label{eq: H_{rc}^n}
\mathcal{H}_{rc}^n(\omega) =
\begin{cases}
k_{rc} \mathcal{C}_{bl}(\omega) + k_{rc} \mathcal{C}_{nl}(\omega), & \text{for } n=1\\
k_{rc} \mathcal{C}_{nl}(n\omega), & \text{for odd} \ n > 1\\
0, & \text{for even} \ n \geqslant 2. 
\end{cases}
\end{equation}
The transfer function for $\mathcal{H}_{bl}$ is defined as
\begin{equation}
\label{eq: Hbl(w)}
\mathcal{H}_{bl}(\omega) =  (1-k_{rc})\mathcal{C}_{bl}(\omega).  
\end{equation}
From \eqref{eq: H_{rc}^n}, \eqref{eq: Hbl(w)}, and the definition of $\mathcal{H}_{p}$ in \eqref{eq: Hpp}, the first-order harmonic of $\mathcal{H}_{p}$ denoted as $\mathcal{H}_{p}^1(\omega)$ is given by
\begin{equation}
\label{eq: Hp1}
\begin{aligned}
\mathcal{H}_{p}^1(\omega) &=  \mathcal{H}_{rc}^1(\omega) + \mathcal{H}_{p}^1(\omega) \\
&=  \mathcal{C}_{bl}(\omega) + k_{rc} \mathcal{C}_{nl}(\omega).
\end{aligned}   
\end{equation}
The higher-order harmonics $n=2k+1>1,\ k\in\mathbb{N}$ of $\mathcal{H}_{p}$ denoted as $\mathcal{H}_{p}^n(\omega)$ is given by
\begin{equation}
\begin{aligned}
\label{eq: Hpn}
\mathcal{H}_{p}^n(\omega) = k_{rc} \mathcal{C}_{nl}(n\omega).
\end{aligned}   
\end{equation}
Combining \eqref{eq: Hp1} and \eqref{eq: Hpn}, the proof is concluded.
\end{proof}
\begin{rem}
Compared to the HOSIDF for the traditional reset controller $\mathcal{H}_{p}$ with $k_{rc}=1$ in \eqref{eq: HOSIDF}, the HOSIDF for $\mathcal{H}_{p}$ in \eqref{eq: H_{p}^n} introduces a new parameter $k_{rc}$ to adjust the nonlinearity $\mathcal{C}_{nl}(n\omega)$ \eqref{rcs:Lbl,Lnl} of the reset controller $\mathcal{C}$, while maintaining the linear element $\mathcal{C}_{bl}(\omega)$ \eqref{eq: RL} in $\mathcal{C}$ unchanged.    
\end{rem}
\begin{cor} (Frequency Response Analysis for the Closed-loop PP-RCS)
\label{thm: Closed-loop HOSIDF for PP-RCS}
Consider a SISO PP-RCS, as illustrated in Fig. \ref{fig1:PP-RC system}, which receives a sinusoidal input signal $r(t) = |R|\sin(\omega t)$, under Assumptions stated in \ref{assum: stable} and \ref{assum:2reset}. The n$^{\text{th}}$ sensitivity function $\mathcal{S}_n(\omega)={E_n(\omega)}/{R_n(\omega)}$ is defined as  
 \begin{equation}
	 	\label{eq: sensitivity functions for PP-RCS}
   \begin{aligned}
	\mathcal{S}_n(\omega)
	   &=\begin{cases}
	 	\frac{1}{1+\mathcal{L}_{o}(\omega)} , & \text{for } n=1\\
	 	-\frac{\Gamma(\omega)\mathcal{L}_{nl}(n\omega)|\mathcal{S}_{1}(\omega)|e^{jn\angle \mathcal{S}_{1}(\omega)}}{1+\mathcal{L}_{bl}(n\omega)} , & \text{for odd} \ n > 1\\
	0, & \text{for even} \ n \geqslant 2 \end{cases}       
   \end{aligned}
	 \end{equation}
  where
 \begin{equation}
\label{eq: Gamma for PP-RCS}
\begin{aligned}
\mathcal{L}_{bl}(n\omega) &= \mathcal{C}_{bl}(n\omega)\mathcal{C}_{2}(n\omega)\mathcal{P}(n\omega),\\
\mathcal{L}_{nl}(n\omega) &= k_{rc}\mathcal{C}_{nl}(n\omega)\mathcal{C}_{2}(n\omega)\mathcal{P}(n\omega),\\
		\end{aligned}
	\end{equation} 
Functions $ \mathcal{C}_{bl}(\omega)$, $ \mathcal{C}_{nl}(\omega)$, and $\Gamma(\omega)$ are given in \eqref{eq: RL} and \eqref{eq: Gamma for PP-RCS}.
\end{cor}
\begin{proof}
As illustrated in Theorem \ref{thm: Open-loop HOSIDF for PP-RCS}, the HOSIDF of a PP-RC is divided into linear element $\mathcal{C}_{bl}(\omega)$ and nonlinear element $k_{rc}\mathcal{C}_{nl}(n\omega)$. The HOSIDF of the open-loop PP-RCS is then given by
\begin{equation}
\label{eq: L_n}
\mathcal{L}_n(\omega) =
\begin{cases}
k_{rc} \mathcal{L}_{bl}(\omega) + k_{rc} \mathcal{L}_{nl}(\omega), & \text{for } n=1\\
k_{rc} \mathcal{L}_{nl}(n\omega), & \text{for odd} \ n > 1\\
0, & \text{for even} \ n \geqslant 2 
\end{cases}    
\end{equation}
where 
\begin{equation}
\label{PP-RCS: Lbl,Lnl}
\begin{aligned}
\mathcal{L}_{bl}(n\omega) &= \mathcal{C}_{bl}(n\omega)\mathcal{C}_{2}(n\omega)\mathcal{P}(n\omega),\\
\mathcal{L}_{nl}(n\omega) &= K_{rc}\mathcal{C}_{nl}(n\omega)\mathcal{C}_{2}(n\omega)\mathcal{P}(n\omega).  
\end{aligned}
\end{equation}
By replacing the open-loop transfer functions for the PP-RCS from \eqref{rcs:Lbl,Lnl} with the new functions \eqref{PP-RCS: Lbl,Lnl}, and substituting them into the closed-loop sensitivity function \eqref{eq: sensitivity functions in CL1}, the proof is concluded.
\end{proof}
\begin{rem}
\label{rem: df, hosidf,et}
Consider a SISO PP-RCS, as illustrated in Fig. \ref{fig1:PP-RC system}, which receives a sinusoidal input signal $r(t) = |R|\sin(\omega t)$, under the assumptions stated in \ref{assum: stable} and \ref{assum:2reset}. 
The steady-state error signal $e(t)$ predicted by the DF in \cite{guo2009frequency} is defined as
    \begin{equation}
    \label{dfeq: y,e,u}
        \begin{aligned}
            e(t) &= \mathscr{F}^{-1}\left[{R_1(\omega)}/{(1+H_1(\omega))}\right].
        \end{aligned}
    \end{equation}  
The steady-state error signal $e(t)$ predicted by the HOSIDF in Corollary \ref{thm: Closed-loop HOSIDF for PP-RCS} is defined as
    \begin{equation}
    \label{hosidfeq: y,e,u}
        \begin{aligned}
            e(t) &= \sum\nolimits_{n=1}^{\infty}\mathscr{F}^{-1}\left[\mathcal{S}_n(\omega)R_n(\omega)\right].
        \end{aligned}
    \end{equation}   
\end{rem}
Due to the nonlinearity inherent in the reset system, the error signal $e(t)$ comprises an infinite number of harmonics. Equation \eqref{dfeq: y,e,u} only addresses the first-order harmonic, whereas equation \eqref{hosidfeq: y,e,u} encompasses higher-order harmonics (for $n>1$). This inclusion results in improved prediction accuracy, as demonstrated in Subsection \ref{subsec: validate hosidf}.

\subsection{Illustrative Example: The Accuracy of the Analysis Method for the PP-RCS}
\label{subsec: validate hosidf}
We employ a mass-spring-damper system (which is a one-degree precision motion stage shown in \cite{karbasizadeh2022band}) as the plant $\mathcal{P}$, whose transfer function is given by:
\begin{equation}
\label{eq:P(s)}
    \mathcal{P}(s) = \frac{6.615e5}{83.57s^2+279.4s+5.837e5}.
\end{equation}

The Proportional-Clegg-Integrator (PCI) controller is built by replacing the integrator in the Proportional-Integrator (PI) controller with the Clegg Integrator (CI). We have constructed a Parallel-Partial PCI-PID (PP-PCI-PID) controller, implemented on $\mathcal{P}(s)$ as depicted in Fig. \ref{fig:PCI_control_system_structure}. The integral frequency in the PCI controller is set to $\omega_i = 2\pi\cdot 15$ rad/s, with the proportional gain value $k_{rc}=0.7$, and the reset value $\gamma=-0.5$. The parameters in the PID structure are chosen to achieve a bandwidth of 150 Hz and a phase margin of 50$\degree$ in open-loop.
\begin{figure}[htp]
	\centering
	\includegraphics[width=0.7\columnwidth]{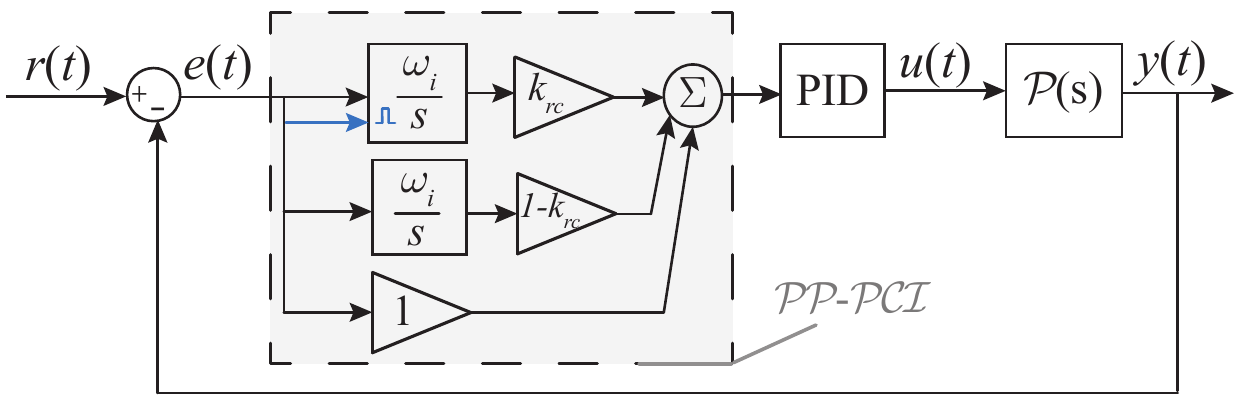}
	\caption{The PP-PCI-PID control system structure.}
	\label{fig:PCI_control_system_structure}
\end{figure}
\begin{figure}[htp]
	\centerline{\includegraphics[width=0.701\columnwidth]{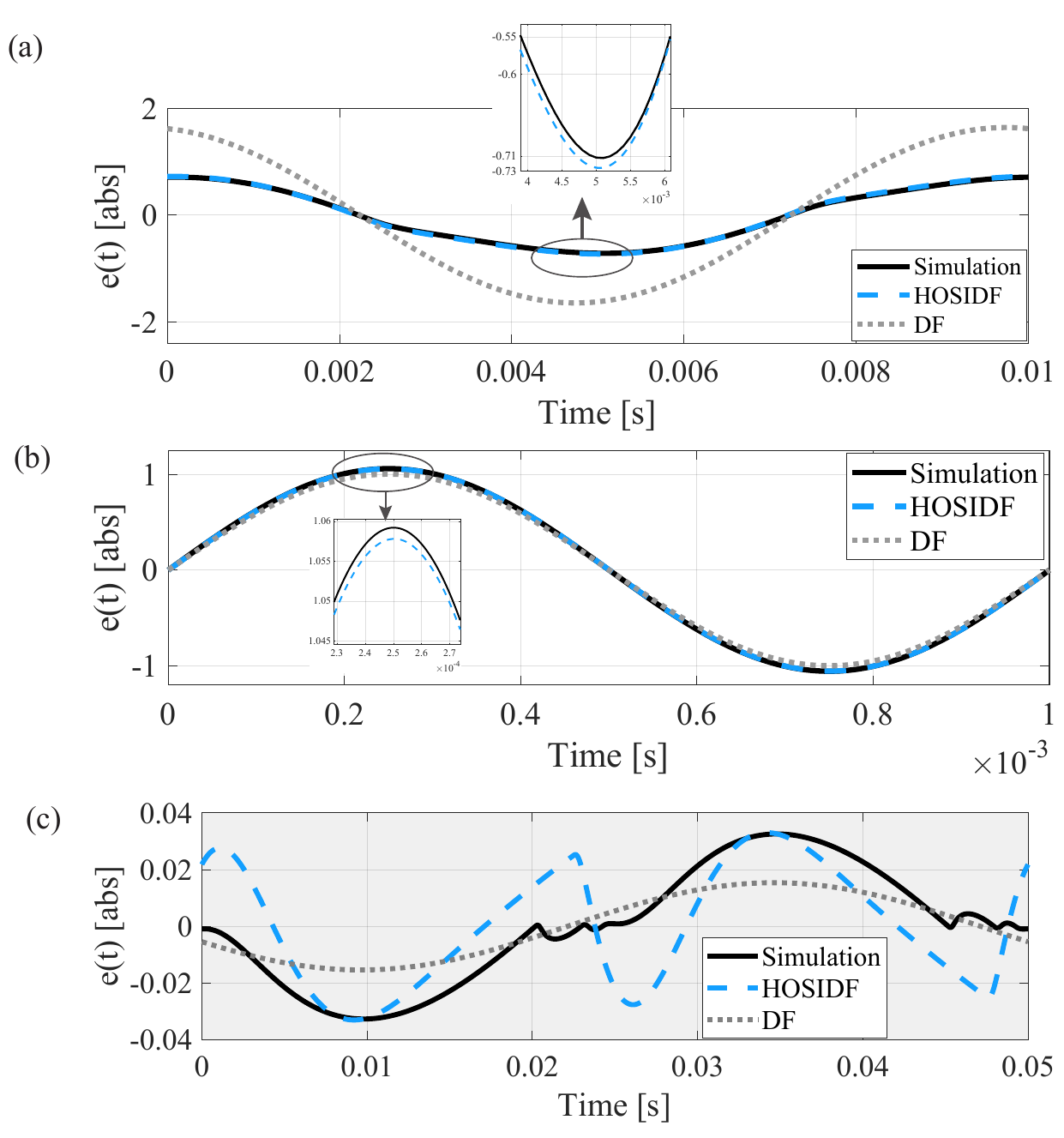}}
	\caption{The steady-state error $e(t)$ of the PP-PCI-PID system at input frequencies of (a) $f = 100$ Hz, (b) $f = 1000$ Hz, and (c) $f=20$ Hz, obtained from simulation, HOSIDF analysis \eqref{hosidfeq: y,e,u}, and DF analysis \eqref{dfeq: y,e,u}.}
	\label{fig: PCIID_n7_krcp5_et_final}
\end{figure}

Figure \ref{fig: PCIID_n7_krcp5_et_final} compares the simulated, equation \eqref{hosidfeq: y,e,u}(HOSIDF)-predicted, and the equation \eqref{dfeq: y,e,u}(DF)-predicted steady-state error $e(t)$ on the PP-PCI-PID system with a sinusoidal input signal $r(t) = \sin(2\pi f t), (\omega = 2\pi f)$ at the input frequencies $f$ of 100 Hz, 1000 Hz, and 20 Hz. From the comparison, two key conclusions can be drawn:

(1) Figure \ref{fig: PCIID_n7_krcp5_et_final}(a) and (b) show that the new HOSIDF method accurately predicts the steady-state error signal of the closed-loop PP-PCI-PID system. There's a significant improvement in accuracy achieved by the HOSIDF analysis compared to the DF analysis, which exceeds 100\% at an input frequency of 100 Hz. The enhanced precision of the new HOSIDF over the DF results from the HOSIDF method incorporating higher-order harmonics into the analysis, whereas the DF method only considers the first-order harmonic in the calculation process, in line with Remark \ref{rem: df, hosidf,et}.

(2) However, in Fig. \ref{fig: PCIID_n7_krcp5_et_final}(c), the HOSIDF inaccurately predicts the steady-state error of the system at 20 Hz due to the violation of Assumption \ref{assum:2reset}. In practical design, it's desirable to adhere to Assumption \ref{assum:2reset} and avoid scenarios with multiple reset instances, as they introduce excessive higher-order harmonics. We apply this example to show that in cases where Assumption \ref{assum:2reset} is not valid, the reliability of the HOSIDF analysis is compromised.
\section{The Performance of the PP-RCS}
\label{sec: result}
\subsection{The PP-CgLp-PI$^2$D control system}
The reset component ``Constant in Gain Lead in Phase" (CgLp) provides phase lead without compromising gain advantages \cite{saikumar2019constant}. We constructed a Parallel-Partial CgLp-PI$^2$D (PP-CgLp-PI$^2$D) control system, as illustrated in Fig. \ref{fig: CgLp_control_system_structure}, where PI$^2$D is the PID controller with an additional integrator. The parameter $\omega_c = 2\pi\cdot 150$ rad/s is utilized. The PI$^2$D is configured to achieve a bandwidth of 150 Hz and a phase margin of 50$\degree$ in open-loop. The transfer function $\mathcal{P}(s)$ is defined in \eqref{eq:P(s)}.
\begin{figure}[htp]
	\centering
	\includegraphics[width=0.7\columnwidth]{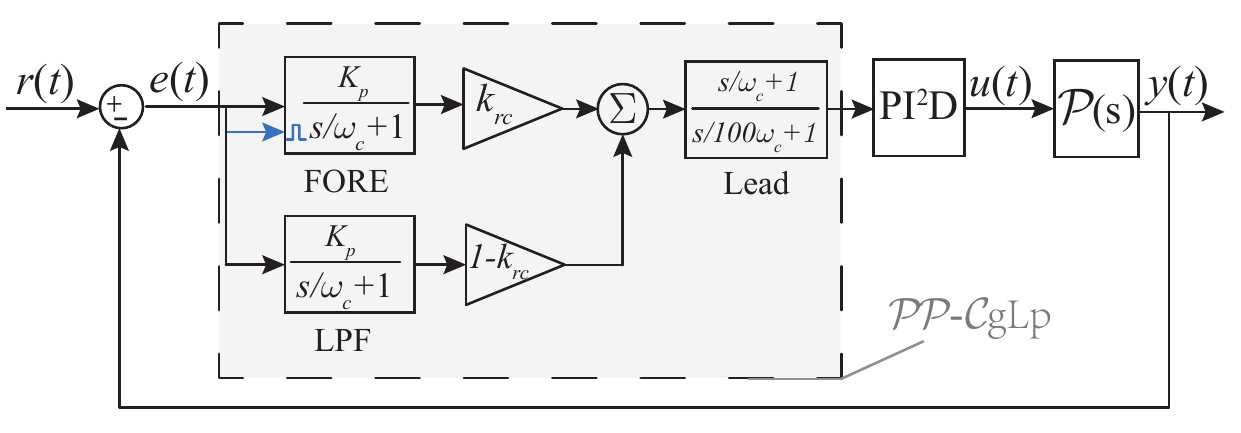}
	\caption{The PP-CgLp-PI$^2$D control system structure.}
	\label{fig: CgLp_control_system_structure}
\end{figure}

Three reset systems denoted as $C_1$, $C_2$, and $C_3$ within the PP-CgLp-PI$^2$D system structure were designed, each with different pairs of $(\gamma, k_{rc})$ values. The $(\gamma, k_{rc})$ parameters for $C_1$-$C_3$ are specified as follows: Linear system $C_1$: $(\gamma, k_{rc})$=$(1, 1)$, traditional reset system $C_2$: $(\gamma, k_{rc})$=$(-0.5, 1)$, and new PP-RCS $C_3$: $(\gamma, k_{rc})$=$(-0.5, 0.5)$. We maintain the same bandwidth of 150 Hz and phase margin of 50$\degree$ in the first-order harmonic for the open-loop systems $C_1$-$C_3$. These three systems under a sinusoidal input signal have approximately two reset instants per steady-state cycle across the entire frequency range, thereby satisfying Assumption \ref{assum:2reset}. Additionally, they have been verified to be stable and convergent, meeting Assumption \ref{assum: stable}.
\subsection{Open-loop Analysis for Systems $C_1$-$C_3$}
Figure \ref{fig: bode_plot_L} illustrates the open-loop HOSIDF $H_p^n(\omega)$ for $n=1$ and $n=3$ in $C_1$-$C_3$. For the first-order harmonics, despite maintaining the same phase margin and bandwidth, the reset control system $C_2$ exhibits slightly larger magnitudes of the first-order harmonics at low frequencies compared to systems $C_1$ and $C_3$. Concerning the higher-order harmonics, system $C_3$ demonstrates lower magnitudes of higher-order harmonics compared to $C_2$.
\begin{figure}[htp]
	\centerline{\includegraphics[width=0.4\textwidth]{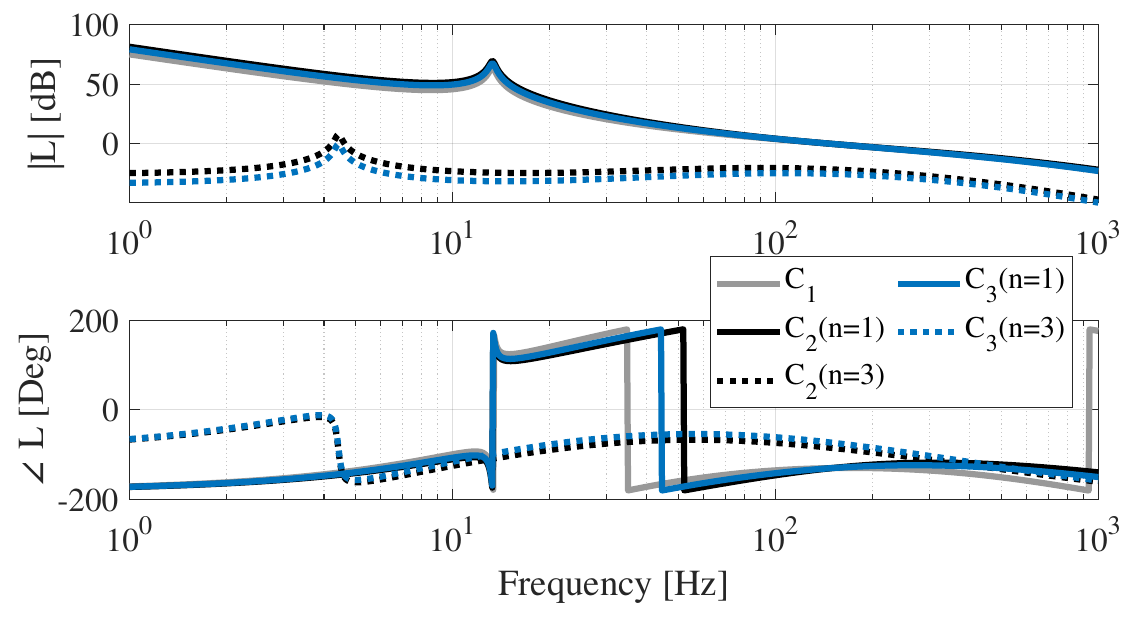}}
	\caption{The HOSIDF $H_p^n(\omega)$ for $n=1$ and $n=3$ in open-loop systems $C_1$-$C_3$.}
	\label{fig: bode_plot_L}
\end{figure}
\subsection{Closed-loop Analysis for Systems $C_1$-$C_3$}
Figure \ref{fig: closed_loop_SSE_rms_infty}(a) and (b) illustrate the $\mathscr{L}_2$ (root mean square $||e||_2$) and $\mathscr{L}_\infty$ ($||e||_\infty$) norms of the steady-state errors in systems $C_1$-$C_3$, predicted by Corollary \ref{thm: Closed-loop HOSIDF for PP-RCS}. Compared to the traditional reset system $C_2$, system $C_3$ exhibits similar performance at low frequencies, achieves better tracking at mid frequencies (ranging from 80 Hz to 150 Hz), with a marginal compromise at high frequencies.
Figure \ref{fig: closed_loop_SSE_rms_infty}(c) and (d) show the simulated steady-state errors at input frequencies of 1 Hz and 100 Hz, validating the prediction results in Fig. \ref{fig: closed_loop_SSE_rms_infty}(a) and (b). 
\begin{figure}[htp]
	\centerline{\includegraphics[width=0.4\textwidth]{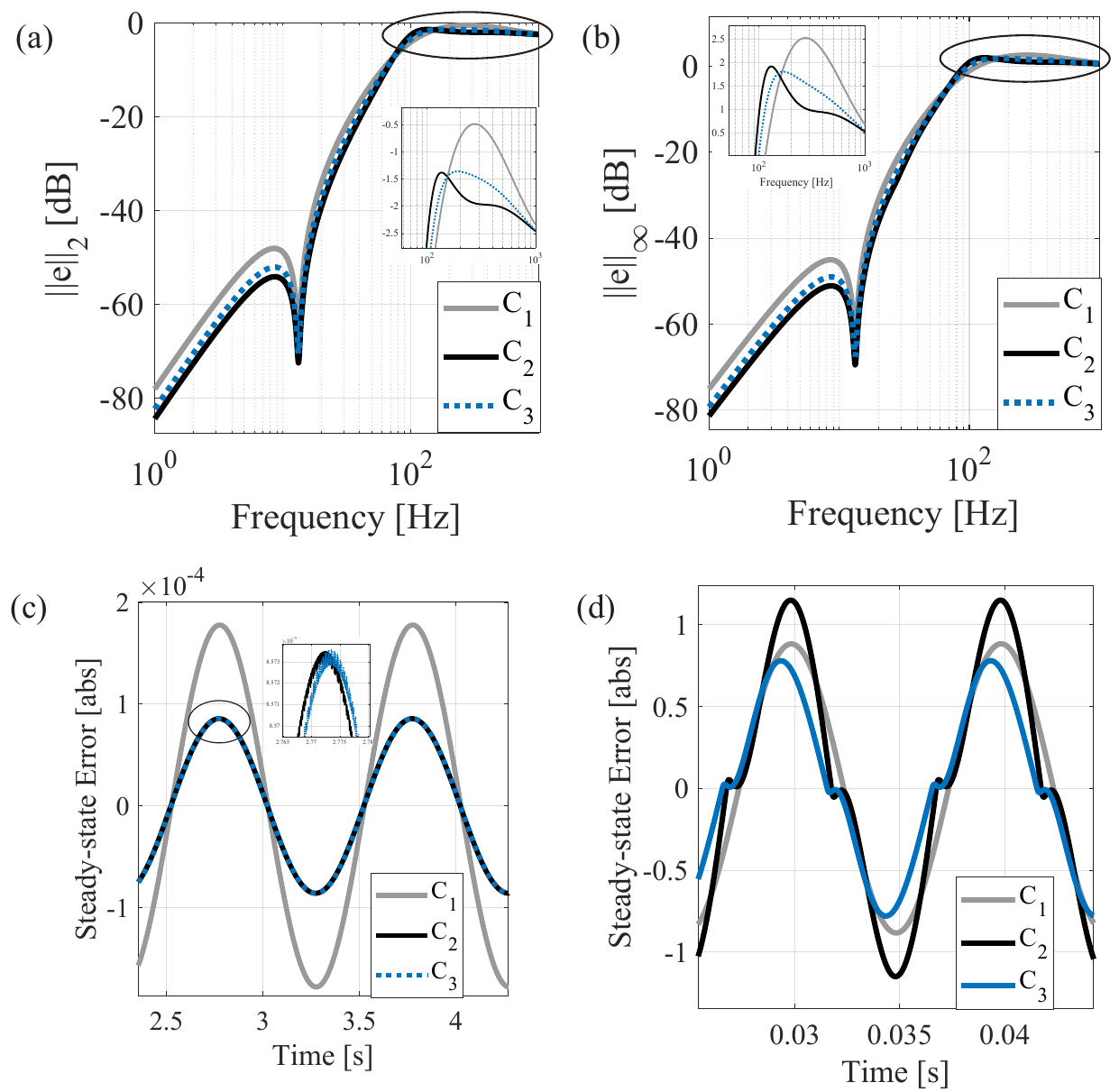}}
	\caption{(a) The $\mathscr{L}_2$ and (b) the $\mathscr{L}_\infty$ norms of steady-state errors in systems $C_1$-$C_3$. The simulated steady-state errors in $C_1$-$C_3$ at input frequencies of (c) 1 Hz and (d) 100 Hz.}
	\label{fig: closed_loop_SSE_rms_infty}
\end{figure}
\subsection{Transient Responses for Systems $C_1$-$C_3$ and Future Work}
Figure \ref{fig: yt_step_final} illustrates the step responses of these three systems. While reset system $C_2$ exhibits lower overshoot compared to the linear system $C_1$, it sacrifices settling time. System $C_3$ achieves the lowest overshoot while maintaining a similar settling time to the linear system $C_1$. The lower overshoot compared to system $C_2$ may be attributed to the fact that in open-loop, the magnitudes of higher-order harmonics of $C_3$ are lower, and in closed-loop, the maximum steady-state error of system $C_3$ is lower (as shown in Fig. \ref{fig: closed_loop_SSE_rms_infty}(b)). Further exploration of the analytical relationship between the transient response and frequency domain analysis is warranted in future research.
\begin{figure}[htp]
	\centerline{\includegraphics[width=0.4\textwidth]{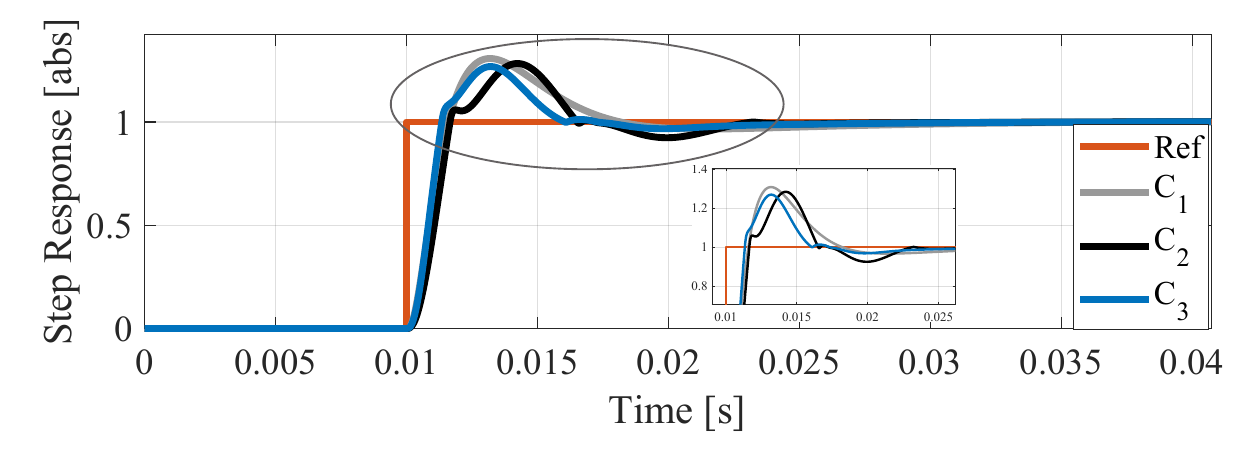}}
	\caption{The step responses for systems $C_1$-$C_3$.}
	\label{fig: yt_step_final}
\end{figure}
\section{Conclusion}
\label{sec: conclusion}
This study introduces a PP-RCS structure, which incorporates a new parameter $k_{rc}$ to adjust the nonlinearity of the reset system while keeping the linearity unchanged. To facilitate effective design in the frequency domain, this work presents frequency response analysis methods for both open-loop and closed-loop PP-RCSs. Simulation results show the enhanced accuracy achieved by the new analysis methods compared to the traditional DF method. An illustrative example on a mass-spring-damper system illustrates the improved transient responses of the PP-RCS compared to traditional reset and linear systems. For future work, experimental validation is warranted. Additionally, it is worthwhile to explore the relation between the transient responses of the PP-RCS and its frequency domain properties.

\bibliographystyle{unsrt}
\bibliography{References}


\end{document}